\documentclass[conference,a4paper]{IEEEtran}

\usepackage{amsthm}
\newtheorem{theorem}{Theorem}
\newtheorem{lemma}{Lemma}
\newtheorem{remark}{Remark}
\newtheorem{definition}{Definition}

\newcommand{\RR}{\mathbb{R}}
\usepackage{cite}
\usepackage{verbatim}

\usepackage[cmex10]{amsmath}
\usepackage{amssymb}

\usepackage[caption=false,font=footnotesize]{subfig}
\usepackage{graphicx}
\usepackage{xcolor}
\usepackage{rotating}

\usepackage{fixltx2e}
\usepackage{algorithmic,algorithm}
\usepackage{pseudocode}

\IEEEoverridecommandlockouts

\begin{document}

\title{Outer Bounds for Multiterminal Source Coding \\ via a Strong Data Processing Inequality}

\author{\IEEEauthorblockN{Thomas A. Courtade}
\IEEEauthorblockA{Department of Electrical Engineering\\
Stanford University\\
Stanford, California, USA\\
Email: courtade@stanford.edu}
\thanks{This work is supported in part by the NSF Center for Science of Information under grant agreement CCF-0939370.}}

\maketitle

\begin{abstract}
 An intuitive outer bound for the multiterminal source coding problem is given.  The proposed bound explicitly couples the rate distortion functions for each source and  correlation measures which derive from a ``strong" data processing inequality.  Unlike many standard outer bounds, the proposed bound is not parameterized by a continuous family of auxiliary random variables, but instead only requires maximizing two ratios of divergences which do not depend on the distortion functions under consideration.
\end{abstract}

\maketitle

\section{Introduction}
We begin with a discussion of the two-encoder quadratic Gaussian source coding problem\footnote{We assume the reader has some familiarity with the multiterminal source coding problem.  For those who are unfamiliar, a formal definition of the problem is given in Section \ref{sec:DefnsResults}.} in order to motivate our main result.  To this end, suppose $X,Y$ are jointly Gaussian -- each with unit variance and correlation $\rho$ -- and distortion is measured under mean square error.  In this setting, 
the set of achievable rate distortion tuples  is given by all $(R_X,R_Y,D_X,D_Y)$ satisfying
\begin{align}
R_X &\geq \frac{1}{2}\log\left(\frac{1}{D_X}\left(1-\rho^2+\rho^22^{-2R_Y}\right)\right)\label{eqn:R1Gauss}\\
R_Y &\geq \frac{1}{2}\log\left(\frac{1}{D_Y}\left(1-\rho^2+\rho^22^{-2R_X}\right)\right)\label{eqn:R2Gauss}\\
R_X+R_Y &\geq \frac{1}{2}\log\left(\frac{(1-\rho^2)\beta(D_X,D_Y)}{2D_XD_Y}\right),\label{eqn:R12Gauss}
\end{align}
where
\begin{align}
\beta(D_X,D_Y)\triangleq1+\sqrt{1+\frac{4\rho^2D_XD_Y}{(1-\rho^2)^2}}.
\end{align}

Long before the converse result was completed in \cite{bib:Wagner2008}, it was known that any $(R_X,R_Y,D_X,D_Y)$ satisfying \eqref{eqn:R1Gauss}-\eqref{eqn:R12Gauss} was achievable.  Indeed, $(R_X,R_Y,D_X,D_Y)$ satisfying \eqref{eqn:R1Gauss}-\eqref{eqn:R12Gauss} correspond to a set of points  in the Berger-Tung achievable region attained by Gausian test channels \cite{bib:BergerLongo1977,bib:Tung1978}.  Moreover, roughly a decade before the sum-rate lower bound \eqref{eqn:R12Gauss} was established in \cite{bib:Wagner2008}, it was proved by Oohama \cite{bib:Oohama1997} that \eqref{eqn:R1Gauss}-\eqref{eqn:R2Gauss} were necessary conditions for $(R_X,R_Y,D_X,D_Y)$ to be achievable.  Thus, in the period  between the publication of \cite{bib:Oohama1997} and \cite{bib:Wagner2008}, ad-hoc lower bounds on the sum-rate could be established as follows.

Noting that the right hand sides of \eqref{eqn:R1Gauss} and \eqref{eqn:R2Gauss} are convex in $R_X$ and $R_Y$, respectively, it is straightforward to establish the necessity of
\begin{align}
R_X +\rho^2R_Y &\geq \frac{1}{2}\log\left(\frac{1}{D_X}\right) \label{eqn:HGR_R1}\\
R_Y +\rho^2R_X &\geq \frac{1}{2}\log\left(\frac{1}{D_Y}\right) \label{eqn:HGR_R2}
\end{align}
in order for $(R_X,R_Y,D_X,D_Y)$ to be achievable.  Indeed, this can be seen by linearizing the RHS of \eqref{eqn:R1Gauss} at $R_Y=0$:
\begin{align}
R_X &\geq \frac{1}{2}\log\left(\frac{1}{D_X}\left(1-\rho^2+\rho^22^{-2R_Y}\right)\right)\\
&\geq \left. \frac{1}{2}\log\left(\frac{1}{D_X}\left(1-\rho^2+\rho^22^{-2R_Y}\right)\right) \right|_{R_Y=0} \notag\\
&+R_Y \cdot \frac{\partial}{\partial R_Y} \left. \frac{1}{2}\log\left(\frac{1}{D_X}\left(1-\rho^2+\rho^22^{-2R_Y}\right)\right) \right|_{R_Y=0}\notag\\
&=\frac{1}{2}\log\left(\frac{1}{D_X}\right) - \rho^2R_Y.
\end{align}
 Thus, a simple sum-rate lower bound in the quadratic Gaussian setting is given by
\begin{align}
R_1+R_2 &\geq \frac{1}{(1+\rho^2)}\left( \frac{1}{2}\log\left(\frac{1}{D_X}\right)+\frac{1}{2}\log\left(\frac{1}{D_Y}\right)\right).\label{eqn:gaussHGR}
\end{align}

In Figure \ref{fig:gauss_p2}, we have compared the lower bound \eqref{eqn:gaussHGR} against the optimal sum-rate constraint  \eqref{eqn:R12Gauss} for $\rho=1/5$.  As evidenced by the figure, the reader will note that the simplified sum-rate lower bound \eqref{eqn:gaussHGR} provides a strikingly tight approximation to \eqref{eqn:R12Gauss}.

In Figure \ref{fig:gauss_p8}, we consider more highly correlated sources with $\rho=4/5$.  As the reader will notice, the accuracy with which \eqref{eqn:gaussHGR} approximates \eqref{eqn:R12Gauss} worsens as $D_XD_Y$ becomes small.  This is to be expected since \eqref{eqn:gaussHGR} was obtained by considering  hyperplanes which support the rate-distortion region when one rate is zero (i.e., in the low-resolution regime).  This situation can be remedied in part by recalling known results for source coding in the high-resolution regime (cf. \cite[Equation (2c)]{bib:ZamirBerger1999}):
\begin{align}
R_X+R_Y & \geq \frac{1}{2}\log\left(\frac{1-\rho^2}{D_XD_Y}\right). \label{eqn:gaussWZ}
\end{align}
Taking the maximum of  \eqref{eqn:gaussHGR} and \eqref{eqn:gaussWZ}  then yields a fairly accurate approximation of \eqref{eqn:R12Gauss}.  The reader should note that \eqref{eqn:gaussWZ} coincides with  the so-called \emph{cooperative lower bound}, in which we assume that both sources are known to a single encoder.  As shown in Figure \ref{fig:gauss_p8}, \eqref{eqn:gaussHGR} can significantly outperform the cooperative lower bound.

\begin{figure}
\centering
\vspace{-0pt}
\includegraphics[trim = 15mm 10mm 10mm 10mm, clip=true,  width=0.47 \textwidth]{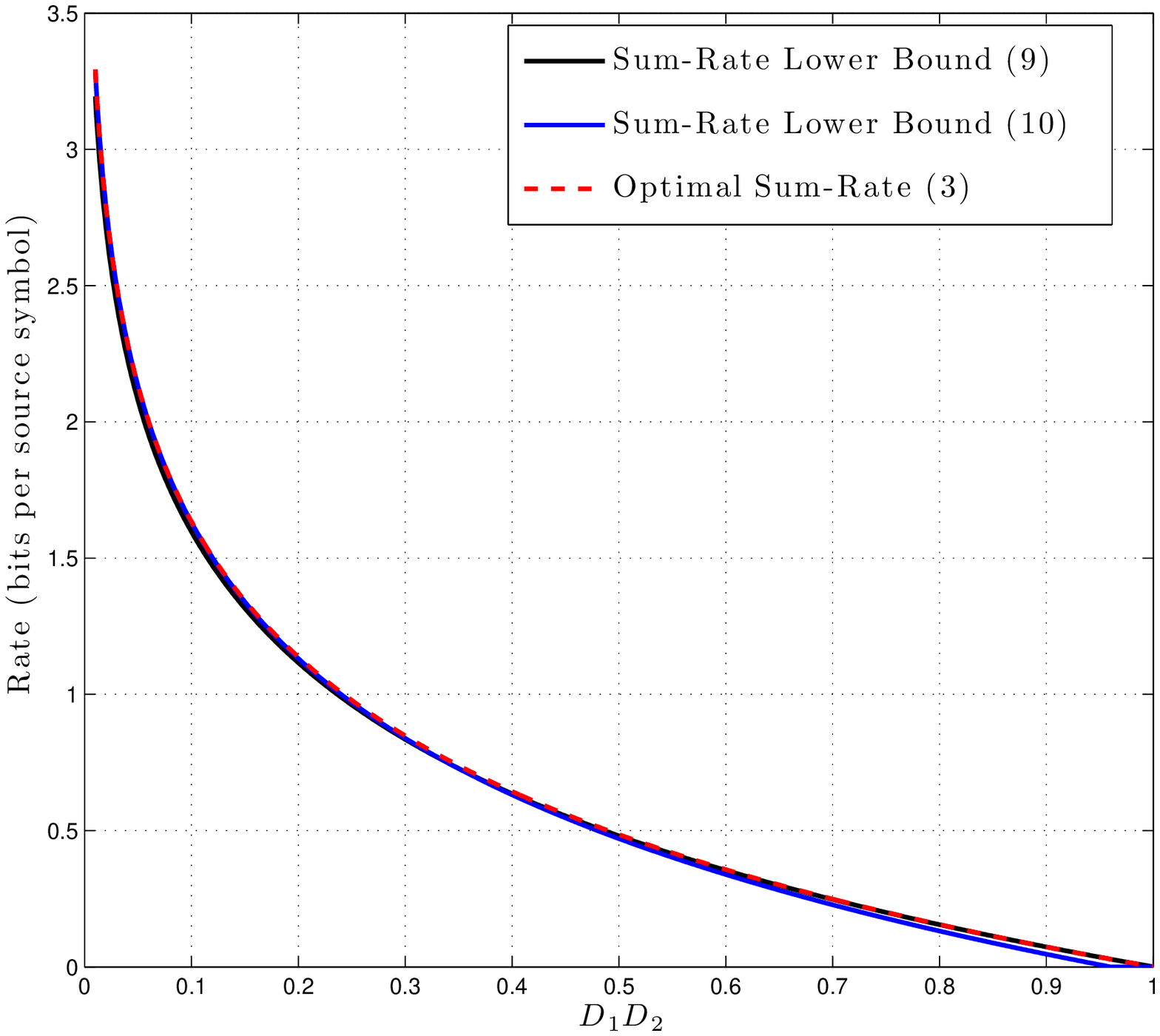}
\caption{Comparison of Eqns. \eqref{eqn:R12Gauss}, \eqref{eqn:gaussHGR}, and \eqref{eqn:gaussWZ} for $\rho=1/5$. }\label{fig:gauss_p2}
\end{figure}

Admittedly, our derivation of \eqref{eqn:gaussHGR} was ad-hoc and required necessity of \eqref{eqn:R1Gauss} and \eqref{eqn:R2Gauss}, which was established by Oohama in \cite{bib:Oohama1997} many years after the multiterminal source coding problem was posed.  Thus, it is desirable to establish a generalization of  \eqref{eqn:HGR_R1} and \eqref{eqn:HGR_R2} to arbitrary sources and distortion measures which does not require known converse results for the specific problem instance under consideration.  This generalization is precisely what we prove in this paper.

Section \ref{sec:DefnsResults} delivers our main result.  Two alternate proofs are given in Section \ref{sec:proofs}, along with a brief discussion.  Section \ref{sec:conc} summarizes our conclusions.  

\section{Definitions and Main Result}\label{sec:DefnsResults}
Throughout this section, let $X,Y$ be random variables with given joint distribution $P_{XY}$.  Let $P_X$ and $P_Y$ denote the marginal distributions of $X$ and $Y$, respectively.  To avoid technicalities, we will assume $\max\{|\mathcal{X}|,|\mathcal{Y}|\}<\infty$.  Without loss of generality, assume $P_X(x)>0$ for all $x\in \mathcal{X}$ and $P_Y(y)>0$ for all $y\in \mathcal{Y}$.

\begin{definition}
Define
\begin{align}
s^*(X;Y) = \sup_{Q_X \neq P_X} \frac{D(Q_Y \| P_Y)}{D(Q_X \| P_X)},
\end{align}
where $Q_Y$ denotes the $y$-marginal distribution of $Q_{XY}=Q_X P_{Y|X}$, and the supremum is over all probability distributions $Q_X$ on $\mathcal{X}$ not identical to $P_X$.
\end{definition}

We remark that $s^*(X;Y) \in [0,1]$ as a consequence of the data processing inequality for relative entropy.%

\begin{figure}
\centering
\includegraphics[trim = 15mm 10mm 10mm 10mm, clip=true, width=0.47 \textwidth]{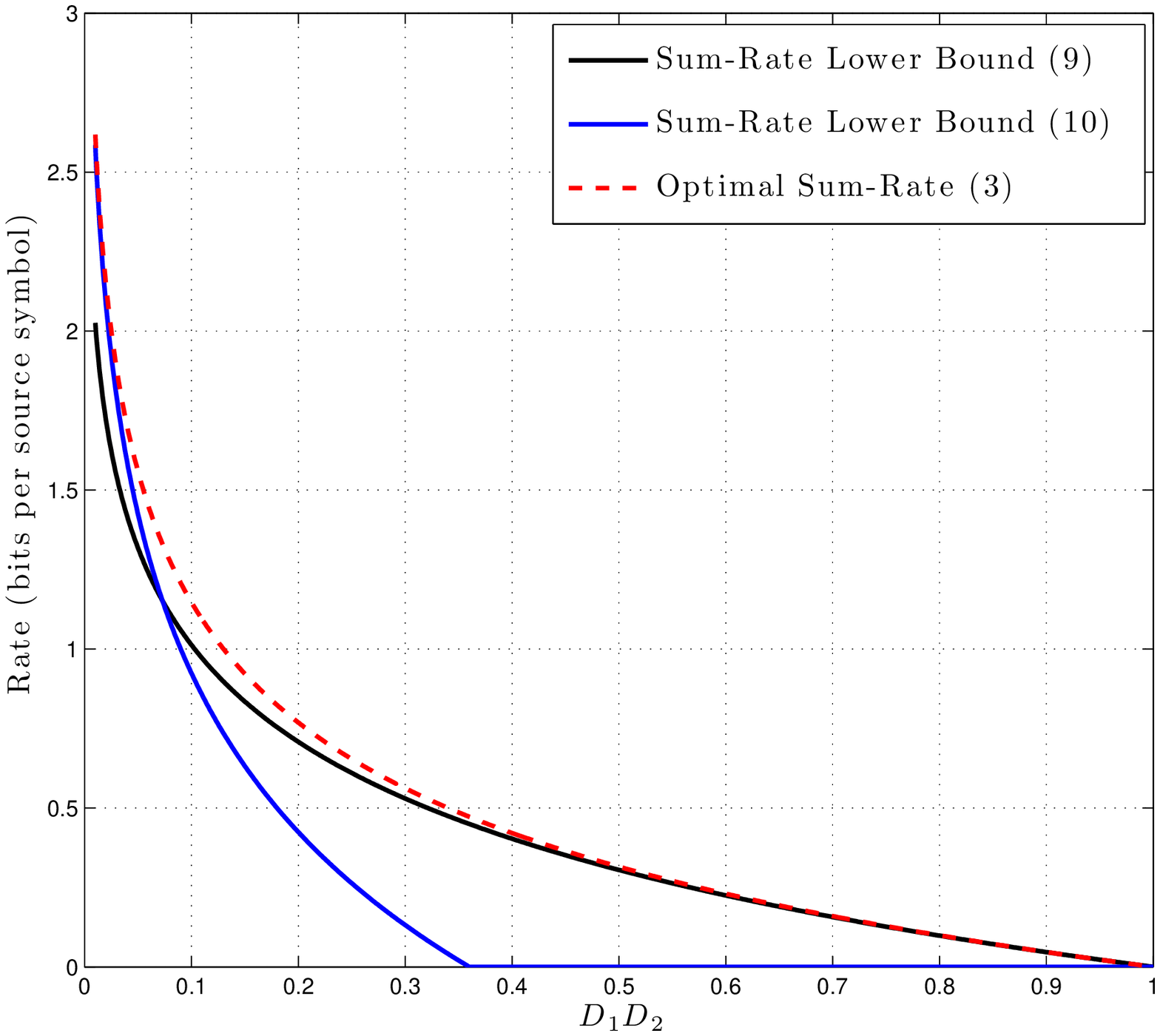}

\caption{Comparison of Eqns. \eqref{eqn:R12Gauss}, \eqref{eqn:gaussHGR}, and \eqref{eqn:gaussWZ} for $\rho=4/5$.}\label{fig:gauss_p8}
\end{figure}

\begin{definition}\label{def:RD}
For a random variable $X$ with alphabet $\mathcal{X}$, a reproduction alphabet $\hat{\mathcal{X}}$, and a distortion function $d_x:\mathcal{X}\times \hat{\mathcal{X}}\rightarrow [0,\infty)$, let $\RR_X(D_X)$ denote the corresponding rate distortion function.  That is,
\begin{align}
\RR_X(D_X) = \min_{p(\hat{x}|x) : \mathbb{E}\left[d_x(X,\hat{X})\right] \leq D_X}I(X;\hat{X}).
\end{align}
\end{definition}
\begin{definition}
Assume $\{X_i,Y_i\}_{i=1}^{\infty}$ is a 2-DMS with joint distribution $P_{XY}$. A rate distortion tuple $(R_X,R_Y,D_X,D_Y)$ is achievable for distortion functions $d_x,d_y$ if, for any $\epsilon>0$, there exists an integer $n$, encoding functions 
\begin{align}
f_x:\mathcal{X}^n &\rightarrow \{1,\dots,2^{nR_X}\}\\
f_y:\mathcal{Y}^n &\rightarrow \{1,\dots,2^{nR_Y}\},
\end{align}
and decoding functions 
\begin{align}
\phi_x:\{1,\dots,2^{nR_X}\}\times \{1,\dots,2^{nR_Y}\}\rightarrow \hat{\mathcal{X}}^n\\
\phi_y:\{1,\dots,2^{nR_X}\}\times \{1,\dots,2^{nR_Y}\}\rightarrow \hat{\mathcal{Y}}^n
\end{align}
which satisfy 
\begin{align}
\mathbb{E}\left[d_x(X^n,\phi_x(f_x(X^n),f_y(Y^n)))\right]&\leq D_X+\epsilon\label{d1}\\
\mathbb{E}\left[d_y(Y^n,\phi_y(f_x(X^n),f_y(Y^n)))\right]&\leq D_Y+\epsilon.\label{d2}
\end{align}
\end{definition}
We remark that distortion between two sequences is defined as the average per-symbol distortion (as usual).
\begin{theorem}\label{thm:HGRbound}
Suppose $(R_X,R_Y,D_X,D_Y)$ is an achievable rate distortion tuple for distortion functions $d_x,d_y$.  Then
\begin{align}
R_X+s^*(Y;X) R_Y &\geq \RR_X(D_X) \label{eqn:thm1}\\
R_Y+s^*(X;Y) R_X &\geq \RR_Y(D_Y).
\end{align}
\end{theorem}
Let $\rho^*(X,Y) \triangleq \max\{ s^*(X;Y),s^*(Y;X)\}$. An immediate corollary of Theorem \ref{thm:HGRbound} is the sum-rate lower bound
\begin{align}
R_X+R_Y \geq \frac{1}{1+\rho^*(X;Y)}\left(\RR_X(D_X)+\RR_Y(D_Y) \right). \label{eqn:HGRgeneral}
\end{align}

We remark that if $X,Y$ are jointly Gaussian with correlation coefficient $\rho$, we have that $\rho^2=\rho^*(X;Y)$ (upon extending the definition of $s^*(X;Y)$ to continuous distributions).  This can be shown by invoking the entropy power inequality in a manner similar to \cite[Section IV-D]{bib:ErkipCover98}.  Thus, Theorem \ref{thm:HGRbound}   generalizes the bounds \eqref{eqn:HGR_R1} and \eqref{eqn:HGR_R2} to any choice of  sources and distortion measures as desired, and \eqref{eqn:HGRgeneral} generalizes \eqref{eqn:gaussHGR}.
\subsection{Discussion}
Roughly speaking, Theorem \ref{thm:HGRbound} implies that, as long as $X,Y$ are not highly correlated under the measures $s^*(X;Y)$ and $s^*(Y;X)$,  compressing with an optimal scheme provides little savings in attainable sum-rate over treating the sources as if they were independent.  For example, consider quaternary sources with joint distribution given by 
\begin{align}
p_{X,Y}(x,y)=\left\{ 
\begin{array}{ll}
1/10 &\mbox{if $x=y$}\\
1/20 &\mbox{if $x\neq y$.}
\end{array}
\right.
\end{align}
By applying the branch and bound algorithm in \cite{bib:Benson2006}, we can compute $\rho^*(X;Y) \approx 0.045$.  Hence, \eqref{eqn:HGRgeneral} implies that separate encoding of $X$ and $Y$ at rates $\RR_X(D_X)$ and $\RR_Y(D_Y)$ incurs at most a $4.3\%$ penalty in sum-rate over an optimal scheme \emph{regardless of which distortion measures are considered}. It does not appear  one can easily make such a claim using previously known results.

Theorem \ref{thm:HGRbound}  has a certain intuitive appeal since it explicitly relates the multiterminal source coding problem to the individual rate distortion functions coupled via the correlation measures $s^*(X;Y)$ and $s^*(Y;X)$.  This tradeoff between correlation and achievable rate-distortion tuples is  obscured in the well-known Berger-Tung outer bound due to its use of auxiliary random variables which often have no physical interpretation (due to the Marokov conditions they satisfy).

Although the Gaussian and quaternary examples we have discussed may give the impression that \eqref{eqn:HGRgeneral} is nearly tight, we point out that this is not always the case.  Indeed, one can devise examples such as $X\sim\mbox{Bernoulli}(1/2)$, $Y=X$ a.s., $d_y\equiv 0$, and $d_x$ equal to Hamming distortion.  In this case  \eqref{eqn:HGRgeneral} is suboptimal by a factor of $2$, however \eqref{eqn:thm1} {is} tight in this case.  Setting aside contrived examples, we believe that Theorem \ref{thm:HGRbound} will give useful bounds for many practical settings of interest (e.g., sensor networks, binaural recording, etc.).

\section{Two Proofs of Theorem \ref{thm:HGRbound}}\label{sec:proofs}
In lieu of proving Theorem \ref{thm:HGRbound}, we shall  prove the stronger result\footnote{Like Theorem \ref{thm:HGRbound}, the outer bound given by Theorem \ref{thm:GeneralHGROB} is efficiently computable.}:
\begin{theorem}\label{thm:GeneralHGROB}
Suppose $(R_X,R_Y,D_X,D_Y)$ is an achievable rate distortion tuple for distortion functions $d_x,d_y$.  Then
\begin{align}
R_X+s^*(Y;X)R_Y &\geq I(X;\hat{X},\hat{Y})\\
R_Y+s^*(X;Y)R_X &\geq I(Y;\hat{X},\hat{Y})
\end{align}
for some conditional distribution $P_{\hat{X},\hat{Y}|X,Y}$ satisfying
\begin{align}
\mathbb{E} d_x(X,\hat{X}) &\leq D_X\\
\mathbb{E} d_y(Y,\hat{Y}) &\leq D_Y.
\end{align}
\end{theorem}

Clearly, Theorem \ref{thm:HGRbound} follows immediately from Theorem \ref{thm:GeneralHGROB} and Definition \ref{def:RD}.    As noted previously,  we shall assume $\max\{|\mathcal{X}|,|\mathcal{Y}|\}<\infty$ to avoid technicalities.

We give two different proofs of Theorem \ref{thm:GeneralHGROB}.  Both arguments rely on the following  ``strong" data processing lemma.
\begin{lemma}[See {\cite{bib:AnanthramEtAl13}}]\label{lem:rhoIneq}
If $U\leftrightarrow X \leftrightarrow Y$ form a Markov chain in that order, then 
\begin{align}
I(Y;U)\leq s^*(X;Y) I(X;U). \label{strongDPI}
\end{align}
\end{lemma}
\begin{remark}
The constant $s^*(X;Y)$ in \eqref{strongDPI} is tight.  Until very recently, it was mistakenly believed that \eqref{strongDPI} held with $s^*(X;Y)$ replaced by $\rho_m^2(X;Y)$  -- the squared Hirschfeld-Gebelein-R\'{e}nyi maximal correlation between $X$ and $Y$ (see 
\cite{bib:ErkipCover98}).  However, it was recently shown in \cite{bib:AnanthramEtAl13} that the correct constant is $s^*(X;Y)$.  We refer the reader to \cite{bib:AnanthramEtAl13} for a detailed discussion.
\end{remark}
\subsection{A Direct Proof of Theorem \ref{thm:GeneralHGROB}}
\begin{proof}[First Proof of Theorem \ref{thm:GeneralHGROB}]
Fix $\epsilon>0$.  Since $(R_X,R_Y,D_X,D_Y)$ is achievable, there exists a $(2^{nR_X},2^{nR_Y},n)$ code $(f_x,f_y,\phi_x,\phi_y)$ which satisfies \eqref{d1} and \eqref{d2}.  In order to simplify notation, we write $\hat{X}^n =  \phi_x(f_x(X^n),f_y(Y^n)))$ and $\hat{Y}^n=\phi_y(f_x(X^n),f_y(Y^n)))$.

With this notation, observe that
\begin{align}
nR_Y &\geq H(f_y(Y^n)) \\
&\geq I(Y^n;f_y(Y^n)|f_x(X^n))\\
&= I(Y^n;f_y(Y^n),f_x(X^n))-I(Y^n;f_x(X^n))\\
&\geq I(Y^n;f_y(Y^n),f_x(X^n))\notag\\
&\quad-s^*(X^n;Y^n)I(X^n;f_x(X^n)) \label{rhoLine}\\
&\geq I(Y^n;\hat{Y}^n,\hat{X}^n)-s^*(X^n;Y^n)nR_X\label{RDline}\\
&= \sum_{i=1}^n I(Y_i;\hat{Y}^n,\hat{X}^n|Y^{i-1})-s^*(X^n;Y^n)nR_X\\
&\geq \sum_{i=1}^n I(Y_i;\hat{Y}_i,\hat{X}_i) -s^*(X^n;Y^n)nR_X \label{eqn:memoryless}\\
&= \sum_{i=1}^n I(Y_i;\hat{Y}_i,\hat{X}_i)-s^*(X;Y)nR_X.\label{singLetter}
\end{align}
In the above string of inequalities, 
\begin{itemize}
\item \eqref{rhoLine} is a consequence of Lemma \ref{lem:rhoIneq} since $f_x(X^n)\leftrightarrow X^n \leftrightarrow Y^n$.
\item \eqref{RDline} follows from  the data processing inequality and the fact that $I(X^n;f_x(X^n)) \leq nR_X$.
\item \eqref{eqn:memoryless} follows by the memoryless property of the source and monotonicity of mutual information.
\item \eqref{singLetter} follows by the tensorization property of $s^*(X^n;Y^n)$ for memoryless sources.  That is, $s^*(X;Y)=s^*(X^n;Y^n)$ (See \cite{bib:AnanthramEtAl13}).
\end{itemize}
Define
\begin{align*}
p(\hat{x},\hat{y}|x,y) = \frac{1}{n}\sum_{i=1}^n \Pr\left(\hat{X}_i=\hat{x},\hat{Y}_i=\hat{y}|X_i=x,Y_i=y\right).
\end{align*}
By linearity of expectation, we have
\begin{align*}
\mathbb{E} d_y(Y,\hat{Y}) &= \mathbb{E}\left[d_y(Y^n,\phi_y(f_x(X^n),f_y(Y^n)))\right]\leq D_Y+\epsilon.
\end{align*}
Since $(X_i,Y_i)$ are identically distributed for all $i$, convexity of mutual information in the conditional distribution implies the desired inequality
\begin{align}
\frac{1}{n}\sum_{i=1}^n I(Y_i;\hat{X}_i,\hat{Y}_i) &\geq I(Y;\hat{X},\hat{Y}).
\end{align}
A symmetric argument completes the proof.
\end{proof}

\subsection{A Proof of Theorem \ref{thm:GeneralHGROB} via Logarithmic Loss}

Interestingly, Theorem \ref{thm:GeneralHGROB} can also be derived from the recent results on source coding under logarithmic loss \cite{bib:CourtadeWeissmanISIT2012}.  This suggests that logarithmic loss may be useful in obtaining other converse results, which are stronger than Theorem \ref{thm:GeneralHGROB}.

Let $\mathcal{M}(\mathcal{X})$ denote the set of probability measures on $\mathcal{X}$.  For $\hat{x}^{(LL)}\in \mathcal{M}(\mathcal{X})$, the logarithmic loss function $d_{LL}:\mathcal{X}\times \mathcal{M}(\mathcal{X})\rightarrow \mathbb{R}$ is defined by
\begin{align}
d_{LL}(x,\hat{x}^{(LL)}) = \log\frac{1}{\hat{x}^{(LL)}(x)},
\end{align}
where $\hat{x}^{(LL)}(x)$ is the probability $\hat{x}^{(LL)}$ assigns to the outcome $x\in \mathcal{X}$.  When $d_x$ and $d_y$ are both logarithmic loss distortion measures (defined for their respective source alphabets $\mathcal{X}$ and $\mathcal{Y}$), the rate distortion region is known.  The characterization of this region is given by the following theorem, which is proved in \cite{bib:CourtadeWeissmanISIT2012}.
\begin{theorem}\label{thm:LL}
$(R_X,R_Y,D_X,D_Y)$ is achievable  under logarithmic loss if and only if
\begin{align}
R_X &\geq I(X;U_X|U_Y,Q) \\
R_Y &\geq I(Y;U_Y|U_X,Q)\\
R_X+R_Y &\geq I(X,Y;U_X,U_Y|Q)\\
D_X &\geq H(X|U_X,U_Y,Q) \\
D_Y &\geq H(Y|U_X,U_Y,Q) 
\end{align}
for some joint distribution of the form $p(x,y)p(q)p(u_X|x,q)p(u_Y|y,q)$ with $|\mathcal{U}_X|\leq |\mathcal{X}|$, $|\mathcal{U}_Y|\leq |\mathcal{Y}|$, and $|\mathcal{Q}|\leq 5$.
\end{theorem}

\begin{proof}[Second Proof of Theorem \ref{thm:GeneralHGROB}]
Since $(R_X,R_Y,D_X,D_Y)$ is achievable, there exists a $(2^{nR_X},2^{nR_Y},n)$ code $(f_x,f_y,\phi_x,\phi_y)$ which satisfies \eqref{d1} and \eqref{d2}. By considering the logarithmic loss reproductions  
\begin{align}
\hat{X}^{(LL)}_i&=\Pr[X_i=x |   {f}_x(X^n),f_y(Y^n)]\\
\hat{Y}^{(LL)}_i&=\Pr[Y_i=y |   {f}_x(X^n),f_y(Y^n)]
\end{align}
 for each index $i=1,2,\dots,n$, Theorem \ref{thm:LL} guarantees the existence of a  joint distribution $p(x,y)p(q)p(u_X|x,q)p(u_Y|y,q)$ with $|\mathcal{U}_X|\leq |\mathcal{X}|$, $|\mathcal{U}_Y|\leq |\mathcal{Y}|$, and $|\mathcal{Q}|\leq 5$ which satisfies\footnote{Establishing the equality in the sum-rate constraint is straightforward.}:
\begin{align}
R_X &\geq I(X;U_X|U_Y,Q) \\
R_Y &\geq I(Y;U_Y|U_X,Q)\\
R_X+R_Y &= I(X,Y;U_X,U_Y|Q)\\
\frac{1}{n}\sum_{i=1}^n H(X_i|  {f}_x(X^n),f_y(Y^n)) &\geq H(X|U_X,U_Y,Q) \label{eqn:LLdist1}\\
\frac{1}{n}\sum_{i=1}^n H(Y_i|  {f}_x(X^n),f_y(Y^n)) &\geq H(Y|U_X,U_Y,Q) .
\end{align}
We now make several observations, from which the claim follows easily.

First, note that \eqref{eqn:LLdist1} is equivalent to
\begin{align}
I(X;U_X,U_Y|Q) &\geq\frac{1}{n}\sum_{i=1}^n I(X_i;  {f}_x(X^n),f_y(Y^n)).  \label{eqn:auxFunc}
\end{align}
Second, since $R_X+R_Y=I(X,Y;U_X,U_Y|Q)$ and $R_X\geq I(X;U_X|U_Y,Q)$, we have
\begin{align}
R_Y &= I(X,Y;U_X,U_Y|Q)-R_X\\
&= I(Y;U_Y|Q)-(R_X-I(X;U_X|U_Y,Q))\\
&\leq I(Y;U_Y|Q). \label{R2upperbound}
\end{align}
Third, we observe that
\begin{align}
&R_X+R_Y = I(X,Y;U_X,U_Y|Q) \\
&= I(X;U_X,U_Y|Q) +I(Y;U_X,U_Y|X,Q) \\
&= I(X;U_X,U_Y|Q) +I(Y;U_Y|Q)-I(X;U_Y|Q)\\
&\geq I(X;U_X,U_Y|Q) +I(Y;U_Y|Q)\notag\\
&\quad-s^*(Y;X) I(Y;U_Y|Q) \label{useDeriv}\\
&=I(X;U_X,U_Y|Q) +(1-s^*(Y;X)) I(Y;U_Y|Q)\\
&\geq I(X;U_X,U_Y|Q) +(1-s^*(Y;X))R_Y, \label{useR2UB}
\end{align}
where \eqref{useDeriv} follows from Lemma \ref{lem:rhoIneq}, and \eqref{useR2UB} follows from \eqref{R2upperbound} and the fact that $s^*(Y;X)\leq 1$.

We rearrange \eqref{useR2UB} and apply \eqref{eqn:auxFunc} to obtain the desired inequality:
\begin{align}
R_X+s^*(Y;X)R_Y &\geq I(X;U_X,U_Y|Q) \\
&\geq \frac{1}{n}\sum_{i=1}^n I(X_i;  {f}_x(X^n),f_y(Y^n))\\
&= \frac{1}{n}\sum_{i=1}^n I(X_i;  \hat{X}^n,\hat{Y}^n)\\
&\geq \frac{1}{n}\sum_{i=1}^n I(X_i;  \hat{X}_i,\hat{Y}_i).
\end{align}
A standard convexity argument (identical to the final step of the alternative proof) completes the argument.
\end{proof}

\subsection{Remarks}

Many applications of strong data processing inequalities begin with a single-letter characterization of the problem of interest.  However, such characterizations are unknown for most multiterminal problems.  Indeed, characterizing the rate-distortion region for the multiterminal source coding problem defined in Section \ref{sec:DefnsResults} for general distortion measures $d_x,d_y$ is a longstanding open problem.  In general, the strong data processing inequality supplied by  Lemma \ref{lem:rhoIneq} can be used in conjunction with the tensorization property of $s^*(X^n;Y^n)$ to obtain meaningful outer bounds in source coding problems without first appealing to a single-letter characterization.  

For instance, a simple sum-rate bound for the CEO problem (cf. \cite{bib:BergerZhangViswanathan1996} for a definition) can be given as follows.  Suppose the observations $(Y_1,\dots,Y_k)$ are conditionally independent given $X$, which should be reproduced at the decoder subject to a constraint on distortion measured under $d_x$.  If $(R_1,\dots,R_k,D)$ is an achievable rate-distortion vector for this CEO problem, then 
\begin{align}
\sum_{i=1}^k s^*(Y_i;X) R_i \geq \RR_X(D_X).
\end{align}

Similar ideas can be applied to non-rate-distortion settings. As an example, consider the problem of generating common randomness:

\begin{definition}
Assume $\{X_i,Y_i\}_{i=1}^{\infty}$ is a 2-DMS with joint distribution $P_{XY}$. A common randomness pair $(C,R)$ is achievable if, for any $\epsilon>0$, there exists an integer $n$, an encoding function $f_m:\mathcal{X}^n \rightarrow \{1,\dots,2^{nR}\}$,
and decoding functions 
\begin{align}
&f_1: K = f_1(X^n), \\
&f_2: K' = f_2(Y^n,f_m(X^n))
\end{align}
which satisfy 
\begin{align}
\Pr(K = K') &> 1-\epsilon \label{eqn:cond1} \\
\frac{1}{n}H(K) &> C - \epsilon \label{eqn:cond2}\\
\frac{1}{n}H(K|K') &< \epsilon. \label{eqn:cond3}
\end{align}
\end{definition}

Let $C(R)$ be the \emph{common randomness capacity}:
\begin{align}
C(R) \triangleq \sup \{ C : (C,R) \mbox{~is achievable.} \}.
\end{align}
In his Ph.D. thesis, Zhao proved the following theorem, which  bounds the maximum number of bits of randomness that can be ``unlocked" by each bit of communication  between users.  
\begin{theorem}[{\cite[Theorem 3]{bib:Zhao2011}}]\label{thm:crCapacity}
\begin{align}
\frac{C(R)}{R} \leq \frac{1}{1-s^*(X;Y)}.
\end{align}
\end{theorem}

Zhao's original proof of Theorem \ref{thm:crCapacity}, while simple, begins with a single-letter characterization of the common randomness capacity $C(R)$, originally due to Ahlswede and Csisz\'{a}r \cite{bib:Ahlswede1998}.  By proceeding along the lines of the direct proof of Theorem \ref{thm:GeneralHGROB}, we can obtain an alternate proof of Theorem \ref{thm:crCapacity} without appealing to a single-letter characterization of $C(R)$.

\begin{remark}
Zhao's statement of Theorem \ref{thm:crCapacity} (i.e., \cite[Theorem 3]{bib:Zhao2011}) involved $\rho_m^2(X,Y)$ instead of $s^*(X;Y)$, and is therefore incorrect in light of \cite{bib:AnanthramEtAl13}.  Above, we give a corrected version.
\end{remark}

\begin{proof}[Proof of Theorem \ref{thm:crCapacity}]
Fix $\epsilon>0$ and consider a scheme which satisfies \eqref{eqn:cond1}-\eqref{eqn:cond3}, with $C = C(R)$.  Then, we have:
\begin{align}
nR &+ ns^*(X;Y)C(R) \geq nR + s^*(X;Y)H(K) \\
&\geq nR + s^*(X;Y)I(K;X^n) \\
&\geq I(f_m(X^n); X^n,K) + I(K;Y^n)\label{eqn:applyMaxCorr} \\
&\geq I(f_m(X^n); X^n,K|Y^n) + I(K;Y^n)\\
&= I(f_m(X^n),Y^n; X^n,K)\\
&\geq I(K';K)\\
&\geq n(C(R)-\epsilon),
\end{align}
where \eqref{eqn:applyMaxCorr} follows from Lemma \ref{lem:rhoIneq} and the tensorization property of $s^*(X^n;Y^n)$.
\end{proof}

\section{Conclusion}\label{sec:conc}
We give an intuitive outer bound for the multiterminal source coding problem which couples the rate distortion functions for each source and the correlation measures $s^*(X;Y),s^*(Y;X)$.  Unlike many standard outer bounds, the proposed bound is not parameterized by a continuous family of auxiliary random variables, but rather only requires evaluation of  $s^*(X;Y)$ and $s^*(Y;X)$.  Roughly speaking, our main result indicates that compressing the sources as if they were independent yields near-optimal sum-rate performance, provided the sources are sufficiently decorrelated in the sense that $\rho^*(X,Y)$ is relatively small.

\bibliographystyle{ieeetr}
\bibliography{bib/mybib}
\end{document}